\documentclass{article}
\usepackage{amsmath,amssymb,amsthm, amsfonts, mathrsfs}
\usepackage{enumerate}
\usepackage[pagewise]{lineno}
\usepackage{calligra,indentfirst,epsfig}
\calligra
\newtheorem{thm}{Theorem}[section]
\newtheorem{lem}{Lemma}[section]
\newtheorem{prop}{Proposition}[section]
\theoremstyle{remark}

\newtheorem{ex}{Example}[section]
\theoremstyle{definition}
\newtheorem{defin}{Definition}[section]
\allowdisplaybreaks
\numberwithin{equation}{section}
 
\begin{document} 
  \title{\bf {On depth spectra of constacyclic codes}}
\author{Tania Sidana {\footnote{Corresponding Author, Email address: taniai@iiitd.ac.in}~\thanks{Research Support by DST-SERB, India, under grant no. EMR/2017/000662, is gratefully acknowledged.} 
}\\
{\it Department of  Mathematics, IIIT-Delhi}\\{\it New Delhi 110020, India}}
\date{}
\maketitle

\begin{abstract}
In this paper, we determine depth spectra of all repeated-root $(\alpha+\gamma \beta)$-constacyclic codes of arbitrary lengths over a finite commutative chain ring $\mathcal{R},$ where $\alpha$ is a non-zero element of the Teichm\"{u}ller set of $\mathcal{R},$ $\gamma$ is a generator of unique maximal ideal of $\mathcal{R}$  and $\beta$ is a unit in $\mathcal{R}.$ We also illustrate our results with some examples.
\end{abstract}
  \section{Introduction}  
The derivative is a well-known operator of  sequences  and is useful in investigating the linear complexity of sequences in game theory, communication theory and cryptography (see \cite{yehuda, blackburn, chan, games}). Etzion \cite{etzio} first applied the derivative operator on codewords of linear codes over finite fields, and defined the depth of a codeword in terms of the derivative operator.  He  showed that there are exactly $k$ distinct non-zero depths attained by non-zero codewords of a $k$-dimensional linear code $\mathcal{C},$  and any $k$ non-zero codewords of $\mathcal{C}$ with distinct depths form a basis of $\mathcal{C}.$ This shows that the depth distribution is an interesting parameter of linear codes.  In the same work, he determined  depth spectra of all binary Hamming codes,  extended binary Hamming codes and first-order binary Reed-Muller codes. He also established a relation between the depth spectrum of a binary linear code of length $2^n$ and the depth spectrum of its dual code. He also showed that the depth of a binary sequence of length $2^n$ as a non-cyclic word is equal to its linear complexity as a cyclic word. Later, Mitchell \cite{mitch} applied the derivative operator on binary sequences (either finite or infinite), and extended the definition of depth for such sequences. He showed that the set of infinite sequences of finite depth corresponds to a set of equivalence classes of rational polynomials, and established an equivalence between infinite sequences of finite depth and sequences of specified periodicity. He also explicitly determined depth spectra of all cyclic codes over arbitrary finite fields. Luo et al. \cite{luo} showed that depth distributions of  linear codes over arbitrary finite fields are completely determined by their depth spectra. They also studied the enumeration problem of counting linear subcodes with a prescribed depth spectrum of a given linear code over a finite field. Using these results, they determined depth distributions of all $r$th order binary Reed-Muller codes. In another related direction,  many important binary non-linear codes are viewed as Gray images of linear codes over the ring $\mathbb{Z}_{4}$ of integers modulo 4 (see \cite{calder,hammons,nech}). Since then, codes over finite commutative chain rings have received a lot of attention. 

Throughout this paper, let $\mathcal{R}$ be a finite commutative chain ring with the unique maximal ideal as $\langle \gamma \rangle.$ The main goal of this paper is to determine depth spectra of all repeated-root $(\alpha+\gamma \beta)$-constacyclic codes of arbitrary lengths over $\mathcal{R},$ where $\alpha $ is a non-zero element of the Teichm\"{u}ller set of $\mathcal{R}$ and $\beta$ is a unit in $\mathcal{R}.$

This paper is structured as follows: In Section \ref{prelim}, we state some preliminaries and derive some basic results that are needed to prove our main results. In Section \ref{sec4}, we determine depth spectra of all repeated-root  $(\alpha+\gamma \beta)$-constacyclic codes of arbitrary lengths over $\mathcal{R}$ (Theorems \ref{Tspec} and \ref{Tspec1}).  In Section \ref{con}, we mention a brief conclusion and discuss some interesting open problems.

  \section{Some preliminaries}\label{prelim}  

Let $R$ be a finite commutative ring with unity,  $N$ be a positive integer, and let $R^N$ be the $R$-module consisting of all $N$-tuples over $R.$  The derivative $D: R^N \rightarrow R^{N-1}$ is defined as $D(a_0,a_1,\cdots,a_{N-1})=(a_1-a_0,a_2-a_1,\cdots,a_{N-1}-a_{N-2})$ for each $(a_0,a_1,\cdots,a_{N-1}) \in R^N.$  
    \begin{defin} \cite{etzio} The depth of a vector $a=(a_0,a_1,\cdots,a_{N-1}) \in R^N,$ denoted by $\text{depth}(a),$ 
  is defined as the smallest integer $i$ (if it exists) satisfying $0 \leq i \leq N-1$ and $D^i(a)=(0,0,\cdots,0) \in R^{N-i}.$ If no such integer $i$ exists (i.e.,  $D^{N-1}(a) \neq 0$),  then  the depth of the vector $a \in R^N$ is defined to be $N.$  
 \end{defin}
 It is easy to see that $\text{depth}(a)=i$ if and only if $D^{i-1}(a)=(b,b,\cdots,b) \in R^{N-i+1}$ for some $b (\neq 0) \in R.$ Further, note that $\text{depth}(a)=0$ if and only if $a=0.$

  \begin{defin}\cite{etzio}  Let $\mathcal{C}$ be a code of length $N$ over $R.$ For $0 \leq \rho \leq N,$ let $\mathcal{D}_{\rho}(\mathcal{C})$ denote the number of codewords in $\mathcal{C}$ having the depth as $\rho.$   The depth distribution of the code $\mathcal{C}$ is defined as the list $\mathcal{D}_{0}(\mathcal{C}), \mathcal{D}_{1}(\mathcal{C}),\cdots,\mathcal{D}_{N}(\mathcal{C}).$ Further, the depth spectrum  of the code $\mathcal{C}$ is defined as $\text{Depth}(\mathcal{C})=\{  i : 1\leq i \leq N \text{ and } \mathcal{D}_{i}(\mathcal{C})\neq 0 \}.$
 \end{defin}
A linear code $\mathcal{C}$ of length $N$ over $R$ is defined as an $R$-submodule of $R^N.$ Further, for a unit $\lambda \in R,$ the code $\mathcal{C}$ is called a $\lambda$-constacyclic code if it satisfies the following: $(a_0,a_1,a_2,\cdots,a_{N-1}) \in \mathcal{C}$ implies that $(\lambda a_{N-1},a_0,a_1,\cdots,a_{N-2}) \in \mathcal{C}.$ Under the standard $R$-module isomorphism  from $ R^N $ onto $R[x]/\langle x^N-\lambda\rangle,$ defined as $(a_0,a_1,\cdots,a_{N-1}) \mapsto a_0+a_1 x+\cdots+a_{N-1}x^{N-1}+\langle x^N-\lambda \rangle$ for each $(a_0,a_1,\cdots,a_{N-1})\in R^N,$ the code $\mathcal{C}$ can be identified as  an ideal of the quotient ring $R[x]/\langle x^N-\lambda\rangle.$ Thus the study of $\lambda$-constacyclic codes of length $N$ over $R$ is equivalent to the study of ideals of the ring $R[x]/\left<x^N-\lambda\right>.$  From now on, we shall represent elements of the ring $R[x]/\langle x^N-\lambda \rangle$ by their representatives in $R[x]$ of degree less than $N,$ and we shall perform their addition and multiplication   modulo $x^N-\lambda.$ Now the derivative of $c(x)  =c_0+c_1x+\cdots+c_{N-1}x^{N-1} \in R[x]/\langle x^N-\lambda \rangle$ is defined as the derivative of the vector $c=(c_0,c_1,\cdots,c_{N-1}) \in R^N.$ In view of this, the depth   of an element $c(x)=c_0+c_1x+\cdots+c_{N-1}x^{N-1} \in R[x]/\langle x^N-\lambda \rangle,$ denoted by $\text{depth}(c(x)),$ is defined as the depth of the vector $c=(c_0,c_1,\cdots,c_{N-1}) \in R^N.$ The following two results are useful in the determination of depths of non-zero codewords of  constacyclic codes.
  \begin{prop}\cite{mitch}\label{p1} Let $0 \leq  i \leq N-1$ be fixed. For $c(x)\in R[x]/\langle x^N-\lambda \rangle,$ let us  write $(1-x)^ic(x)=d_0+d_1x+d_2x^2+\cdots+d_{N-1}x^{N-1}$ modulo $x^N-\lambda.$ Then the $i$th  derivative $D^i(c(x))$ of the element $c(x) \in R[x]/\langle x^N-\lambda\rangle$ is given by   $$D^i(c(x))~=~(d_i,d_{i+1},\cdots, d_{N-1}),$$ i.e., $D^i(c(x))$  appears as the last $N-i$ coefficients of the polynomial $(1-x)^i c(x)$ modulo $x^N-\lambda.$
\end{prop}  

  \begin{lem}\label{ll3} Let $c(x) \in R[x]/\langle x^N-\lambda \rangle,$ and let $\ell,t$ be positive integers satisfying   $\ell +t \leq N.$ If $\text{depth}((1-x)^\ell c(x))=t,$ then  $\text{depth}(c(x))=\ell+t.$
\end{lem}
\begin{proof}   As $\text{depth}((1-x)^\ell c(x))=t,$ we have $D^{t-1}((1-x)^\ell c(x))=(d,d,\cdots,d)\in R^{N-t+1}$ for some $d(\neq 0)\in R.$ Now by Proposition \ref{p1}, we see that the last $(N-t+1)$ coefficients of the element $(1-x)^{t-1}(1-x)^\ell c(x) \in R[x]/\langle x^N-\lambda \rangle$ are equal to $d.$ In particular, the last $(N-t-\ell+1)$  coefficients of the element $(1-x)^{\ell+t-1}c(x)\in R[x]/\langle x^N-\lambda \rangle$ are equal to $d.$ By Proposition \ref{p1} again, we note that $D^{t+\ell-1}(c(x))=(d,d,\cdots,d) \in R^{N-t-\ell+1},$ which gives $\text{depth}(c(x))=\ell+t.$ This proves the lemma.
  \end{proof}  
However,  when $\ell+t > N,$  Lemma \ref{ll3} does not hold. In this case, the depth of $c(x)$ may be strictly less than $N.$ The following example illustrates this.
 \begin{ex} Let $R=\mathbb{Z}_4,$  $N=4$ and $\lambda=-1.$
Let us take  $c(x)=x+2x^2+3x^3\in \mathbb{Z}_4[x]/\langle x^4+1\rangle,$ and $c_1(x)=(1-x)c(x)=3+x+x^2+x^3 \in \mathbb{Z}_4[x]/\langle x^4+1\rangle.$ It is easy to see that $t=\text{depth}(c_1(x))=4$ and $\text{depth}(c(x))=3.$ Here we note that $\ell+t=1+t=5>4=N$ and $\text{depth}(c(x))=3<4=N.$  \end{ex}

The following proposition plays a key role in the determination of  depth spectra of linear codes over finite fields.

  \begin{prop} \cite{etzio}\label{p4}  If $\mathcal{C}$ is a linear code over a finite field, then $|\text{Depth}(\mathcal{C})|$ equals the dimension of $\mathcal{C}.$ (Throughout this paper, $|A|$ denotes the cardinality of the set $A.$)
 \end{prop}  

In the following proposition, depth spectra of all cyclic codes over finite fields are determined.
  \begin{prop} \cite{mitch}\label{p2} Let $\mathcal{C}$ be a cyclic code of length $N$ over the finite field $\mathbb{F}_{q}$  with the generator polynomial as $g(x).$ Then  for an integer $t \geq 0,$ $(x-1)^t || \frac{x^N-1}{g(x)}$ in $\mathbb{F}_{q}[x]$ if and only if $\text{Depth} (\mathcal{C})=\{1,2,\cdots,t\} \cup \{\text{deg }g(x)+t+1,\text{deg }g(x)+t+2,\cdots, N-1,N\}.$ Here by $(x-1)^t || \frac{x^N-1}{g(x)}$ in $\mathbb{F}_{q}[x],$ we mean $(x-1)^t | \frac{x^N-1}{g(x)}$ and $(x-1)^{t+1}\nmid \frac{x^N-1}{g(x)}$ in $\mathbb{F}_{q}[x].$ (Throughout this paper, $\text{deg }h(x)$ denotes the degree of a non-zero polynomial $h(x) \in \mathbb{F}_{q}[x].$)
\end{prop} 
 In a recent work, Zhang \cite[Th. 4]{zhang} determined depth spectra of all $\eta$-constacyclic codes over finite fields of prime order when $\eta \neq 1.$ Working in a similar way, this result can be extended  to $\eta$-constacyclic codes over arbitrary finite fields, which we state as follows: 
   \begin{thm}\label{p3} Let $\eta (\neq 1)$ be a non-zero element of the finite field  $\mathbb{F}_{q}$ of order $q.$  Let $\mathcal{C}$ be a non-trivial  $\eta$-constacyclic code of length $N$ over $\mathbb{F}_{q}$ with the generator polynomial as $g(x).$ Then we have   $$\text{Depth}(\mathcal{C})~=~\{\text{deg }g(x)+1,\text{deg }g(x)+2,\cdots, N-1,N\}.$$
\end{thm}
  \begin{proof} Working in a similar manner as in Theorem 4 of Zhang \cite{zhang}, the desired result follows. 
 \end{proof}  
From now on, throughout this paper,  let $\mathcal{R}$ be a finite commutative chain ring  with unity, and let  $\gamma$ be a generator of the maximal ideal of $\mathcal{R}.$ Further, let $e $ be the nilpotency index of $\gamma,$ and let  $\overline{\mathcal{R}}=\mathcal{R}/\langle\gamma \rangle$ be the residue field of $\mathcal{R}.$  As $\overline{\mathcal{R}}$ is a finite field, we assume that  $\overline{\mathcal{R}} \simeq  \mathbb{F}_{p^m}$ for some prime $p$ and positive integer $m,$ where $\mathbb{F}_{p^m}$ is the finite field of order $p^m.$ Further, there exists an element $\zeta \in \mathcal{R}$ whose multiplicative order is $p^m-1.$  The set $\mathcal{T}=\{0,1,\zeta, \cdots, \zeta^{p^m-2}\}$ is called the Teichm\"{u}ller set of $\mathcal{R}.$ Let $^{-}: \mathcal{R} \rightarrow \overline{\mathcal{R}}$ be the natural epimorphism from $\mathcal{R}$ onto $\overline{\mathcal{R}},$ which is given by $r \mapsto \overline{r}=r +\langle\gamma \rangle$ for each $r \in \mathcal{R}.$  For a unit $\lambda \in \mathcal{R},$ the map $^{-}$ can be further extended to a map $\mu$ from $\mathcal{R}_{\lambda}=\mathcal{R}[x]/\langle x^N-\lambda \rangle$ into $\overline{\mathcal{R}}_{\lambda} = \overline{\mathcal{R}}[x]/ \langle x^{N}  -\overline{\lambda} \rangle$ as follows:
  \begin{equation*} 
\sum\limits_{i=0}^{N-1}a_ix^i  ~\mapsto~ \sum\limits_{i=0}^{N-1}\overline{a_i}x^i \text{~~~for each~~~} \sum\limits_{i=0}^{N-1}a_ix^i \in\mathcal{R}_{\lambda}.  \end{equation*}
 It is easy to observe that $\mu$ is a surjective ring homomorphism from $\mathcal{R}_{\lambda}$ onto $\overline{\mathcal{R}}_{\lambda}.$

  \begin{prop}  \label{teich}\cite{mcdon} The following hold.\begin{enumerate}  \item[(a)] The characteristic of  $\mathcal{R}$ is $p^a,$ where $1 \leq a \leq e.$ Moreover, we have $|\mathcal{R}|=|\overline{\mathcal{R}}|^{e}=p^{me}.$
  \item[(b)] 
 For a positive integer $s$ and a non-zero $\theta \in \mathcal{T},$ there exists $\theta_0 \in \mathcal{T}$ satisfying $\theta_0^{p^s}=\theta.$
  \item[(c)] Each element $r \in \mathcal{R}$ can be uniquely expressed as $r=r_0+ r_1 \gamma+ r_2 \gamma^2+\cdots + r_{e-1} \gamma^{e-1},$ where $r_i \in \mathcal{T}$ for $0 \leq i \leq e-1.$  Moreover, $r$ is a unit in $\mathcal{R}$ if and only if $r_0 \neq 0.$  
\end{enumerate}\end{prop} 

 By Proposition \ref{teich}(c), we see that a unit $\lambda \in \mathcal{R}$ can be written as $\lambda= \alpha+\gamma \beta,$ where $\alpha (\neq 0) \in \mathcal{T}$ and $\beta \in \{0\}\cup (\mathcal{R}\setminus \langle \gamma^{e-1}\rangle).$ Let $\mathcal{C}$ be a $\lambda$-constacyclic code  of length $N$ over $\mathcal{R},$ (i.e., an ideal of the ring $\mathcal{R}_{\lambda}$). For $0 \leq i \leq e-1,$ the $i$th torsion code of $\mathcal{C}$ is defined as
  \begin{equation*}
 \text{Tor}_i(\mathcal{C})~=~\{\mu(f(x))\in \overline{\mathcal{R}}_{\lambda}  | \gamma^i f(x) \in \mathcal{C}  \}.
  \end{equation*}  
  \begin{thm}\cite{norton}\label{card} If $\mathcal{C}$ is a $\lambda$-constacyclic code of length $N$ over $\mathcal{R},$ then we have $|\mathcal{C}|=\prod\limits_{i=0}^{e-1} |\text{Tor}_i(\mathcal{C})|.$
  \end{thm}

From this point on, let $ \lambda=\alpha+\gamma \beta,$ where $\alpha(\neq 0) \in \mathcal{T}$ and $\beta $ is  a unit in $\mathcal{R}.$ To determine all $\lambda$-constacyclic codes of length $np^s$ over $\mathcal{R}$ and their Torsion codes,  we see, by Proposition \ref{teich}(b), that there exists $\alpha_0 \in \mathcal{T}$ satisfying $\alpha_0^{p^s}=\alpha.$ Further, as  $\gcd(n,p)=1,$ by Theorem 2.7 of  Norton and S$\breve{a}$l$\breve{a}$gean \cite{norton}, we can write $x^n-\alpha_0=f_1(x)f_2(x)\cdots f_r(x),$ where  $f_1(x),f_2(x),\cdots, f_r(x)$ are monic basic irreducible pairwise coprime polynomials in $\mathcal{R}[x].$
In the following theorem, we determine all $\lambda$-constacyclic codes of length $np^s$ over $\mathcal{R}$ and their Torsion codes.
  \begin{thm}  \label{Tstruc} Let $\mathcal{C}$ be a $\lambda$-constacyclic code of length $np^s$ over $\mathcal{R},$ (i.e., an ideal of the ring $\mathcal{R}_{\lambda}$). Then we have the following: 
\begin{enumerate}  \item[(a)] $\mathcal{C}=\langle \prod\limits_{\ell=1}^{r} f_\ell(x)^{k_\ell} \rangle$ in $\mathcal{R}_\lambda,$ where $0 \leq k_\ell \leq ep^s$ for $1 \leq \ell \leq r.$ 

  \item[(b)]  For $0 \leq i \leq e-1,$ we have    \begin{equation*} \text{Tor}_i(\mathcal{C})~=~\bigg\langle\prod\limits_{\ell=1}^{r}\overline{ f_\ell(x)}^{\tau_\ell(i)}\bigg\rangle \text{~~~in~~~} \overline{\mathcal{R}}_\lambda, \end{equation*} where $\tau_\ell(i)=\min\{(i+1)p^s, k_\ell\}-\min\{ i p^s, k_\ell\}$ for $ 1 \leq \ell \leq r.$ 
 \end{enumerate}\end{thm}
 
 \begin{proof} \begin{enumerate}   \item[(a)] Working in a similar manner as  in Theorem 3.1 of Sharma and Sidana  \cite{sharma1} and by applying  the Chinese Remainder Theorem, the desired result follows.
  \item[(b)] Working in a similar manner as in Theorem 3.5 of \cite{zhu}, the desired result follows.\end{enumerate}   \end{proof} 

 Note that each non-zero element $c(x)\in \mathcal{R}_\lambda$ can be expressed as $c(x)=\gamma^\ell A(x),$ where $0 \leq \ell \leq e-1$ and $A(x)  \in \mathcal{R}_{\lambda}$ satisfies $\mu(A(x)) \neq 0.$ In the following lemma, we relate the depth of $c(x)$ with the depth of $\mu(A(x)).$
  \begin{lem}\label{l2} Let $c(x)$ be a non-zero element of $\mathcal{R}_\lambda.$ Let us write $c(x)=\gamma^\ell A(x),$ where $0 \leq \ell \leq e-1$ and $\mu(A(x)) \neq 0.$ Then the following hold.
\begin{enumerate} \item[(a)] We have $\text{depth}(c(x)) \geq \text{depth} (\mu(A(x))).$ \item[(b)] When $\ell=e-1,$ we have $\text{depth}(c(x)) = \text{depth} (\mu(A(x))).$ \end{enumerate}
\end{lem}
\begin{proof} 
(a) When $\text{depth}(c(x))=N,$ the result holds trivially.  Now we assume that $\text{depth}(c(x))=t  < N.$ This implies that $\gamma^\ell  D^{t}(A(x))=D^{t}(\gamma^\ell A(x))=D^{t}(c(x))=(0,0,\cdots,0)\in \mathcal{R}^{N-t},$ which further implies that $D^{t}( A(x)) \in \langle \gamma^{e-\ell}\rangle^{N-t}.$ From this, we obtain $D^{t}(\mu( A(x)))=\mu(D^{t}( A(x)))=(0,0,\cdots,0) \in \overline{\mathcal{R}}^{N-t}.$ This shows that $\text{depth}(\mu(A(x))) \leq t=\text{depth}(c(x)).$  

(b) If $\text{depth} (\mu(A(x)))=N,$ then by part (a), we get $N\geq \text{depth}(c(x)) \geq \text{depth} (\mu(A(x)))=N,$ which gives $\text{depth}(c(x))=N= \text{depth} (\mu(A(x))).$  Now we assume that $\text{depth} (\mu(A(x)))=k < N.$ This gives $\mu(D^k(A(x)))=D^k(\mu(A(x)))=(0,0,\cdots,0) \in \overline{\mathcal{R}}^{N-k},$ which implies that $D^k(A(x))\in \langle \gamma\rangle^{N-k}.$ This further implies that $D^k(c(x))=D^k(\gamma^{e-1} A(x))=\gamma^{e-1}D^k(A(x))=(0,0,\cdots,0) \in \mathcal{R}^{N-k}.$ This shows that $\text{depth}(c(x)) \leq k=\text{depth} (\mu(A(x))).$ From this and  by part (a), we get the desired result.
  \end{proof}

From now on, we will follow the same notations as in Section \ref{prelim}.

  \section{Determination of depth spectra of $\lambda$-constacyclic codes of length $np^s$ over $\mathcal{R}$}\label{sec4}  
In this section, we shall  determine depth spectra of all $\lambda$-constacyclic codes of length $np^s$ over $\mathcal{R}.$  Towards this,  we first prove the following lemma.
  \begin{lem} \label{fac1} Let $A(x) \in \mathcal{R}[x]$ be such that $\mu(A(x)) \neq 0.$ If there exists an integer $t$  satisfying $0 \leq t \leq e-1$ such that the polynomial $x^{np^s}-\lambda$ divides $\gamma^t A(x)$ in $\mathcal{R}[x], $ then $x^{np^s}- \overline{\lambda}$ divides $\mu(A(x))$ in $\overline{\mathcal{R}}[x].$
  \end{lem}
\begin{proof} As $x^{np^s}-\lambda$ divides $\gamma^t A(x)$ in $\mathcal{R}[x],$ we can write $\gamma^t A(x)=(x^{np^s}-\lambda )B(x),$ where $B(x) \in \mathcal{R}[x].$ Now we observe that all the  coefficients of the polynomial  $B(x)$ lie in  $\langle \gamma^t\rangle.$ So we can  write $B(x)=\gamma^t V(x),$ where $V(x) \in \mathcal{R}[x].$ This gives $\gamma^t A(x)=\gamma^t(x^{np^s}-\lambda ) V(x),$ which implies that $\gamma^{e-1} A(x)=\gamma^{e-1}(x^{np^s}-\lambda ) V(x).$ As $\mu(A(x))\neq 0,$ we have $\mu(V(x))\neq0.$  Next we see that   \begin{eqnarray*}\gamma^{e-1}(x^{np^s}-\lambda)&=&\gamma^{e-1}\big((x^{n}-\alpha_0+\alpha_0)^{p^s}-\alpha_0^{p^s}-\gamma\beta\big)\\&=&\gamma^{e-1}(x^n-\alpha_0)^{p^s} +\gamma^{e-1}\sum\limits_{k=1}^{p^s-1} \binom{p^s}{k}(x^n-\alpha_0)^{k}{\alpha_0}^{p^s-k}.  \end{eqnarray*}
Further, for $1\leq k \leq p^s-1,$ by applying Kummer's Theorem, we note that $p$ divides $\binom{p^s}{k},$ which implies that $ \binom{p^s}{k}\in\langle\gamma\rangle.$ From this, we obtain $\gamma^{e-1} A(x)=\gamma^{e-1}(x^{n}-\alpha_0)^{p^s}V(x),$ which gives $\gamma^{e-1} \big(A(x)-(x^{n}-\alpha_0)^{p^s}V(x)\big)=0.$ From this, it follows that    $$\mu(A(x))~=~(x^{n}-\overline{\alpha}_0)^{p^s}\mu(V(x))~=~(x^{np^s}-\overline{\lambda})\mu(V(x))\text{~~~in~~~}\overline{\mathcal{R}}[x],  $$ which proves the lemma.
  \end{proof}

Next  by Theorem \ref{Tstruc}(a), we recall that a $\lambda$-constacyclic code $\mathcal{C}$ of length $np^s$ over $\mathcal{R}$ is generated by $ \prod\limits_{\ell=1}^{r} f_\ell(x)^{k_\ell},$  where $0 \leq k_\ell \leq ep^s$ for $1 \leq \ell \leq r.$ Further, for $0 \leq i \leq e-1,$ by Theorem \ref{Tstruc}(b), we note that $\text{Tor}_i(\mathcal{C})=\langle \prod\limits_{\ell=1}^{r}\overline{ f_\ell(x)}^{\tau_\ell(i)} \rangle,$ where $\tau_\ell(i)=\min\{(i+1)p^s, k_\ell\}-\min\{ i p^s, k_\ell\}$ for $ 1 \leq \ell \leq r.$ We also recall that $\text{deg }f_{\ell}(x)=d_{\ell}$ for $1 \leq \ell \leq r.$ Let us define
  $$\mathcal{S}_{1}(\mathcal{C})~=~\sum \limits_{\ell=1}^{r}d_\ell \tau_\ell(e-1)\text{~~ and~~ }\mathcal{S}_{2}(\mathcal{C})~=~\sum \limits_{\ell=2}^{r}d_\ell \tau_\ell(e-1).$$

Now we shall distinguish the following two cases: (i)  $\overline{\lambda} \neq 1$ and (ii) $\overline{\lambda} =1.$

In the following theorem, we consider the case $\overline{\lambda} \neq 1,$ and we determine depth spectra of all $\lambda$-constacyclic codes of length $np^s$ over $\mathcal{R}.$ 
  \begin{thm}\label{Tspec} Let $\mathcal{C} =\langle \prod\limits_{\ell=1}^{r} f_\ell(x)^{k_\ell} \rangle$ be a non-trivial $\lambda$-constacyclic code of length $np^s$ over $\mathcal{R}$ with the $i$th torsion code as $\text{Tor}_i(\mathcal{C})=\langle\prod\limits_{\ell=1}^{r}\overline{ f_\ell(x)}^{\tau_\ell(i)}\rangle$ for $0 \leq i \leq e-1,$ where $0 \leq k_\ell \leq ep^s$ and $\tau_\ell(i)=\min\{(i+1)p^s, k_\ell\}-\min\{ i p^s, k_\ell\}$ for each $i$ and $\ell.$  
When $\overline{\lambda} \neq 1,$  the depth spectrum of the code $\mathcal{C}$ is given by 
  $$ \text{Depth}(\mathcal{C})~=~\{\mathcal{S}_{1}(\mathcal{C}) +1,\mathcal{S}_{1}(\mathcal{C})+2,\cdots, np^s\}.  $$
\end{thm}
\begin{proof} To prove the result, by Theorem \ref{Tstruc}(b), we see that $\text{Tor}_{e-1}(\mathcal{C})$ is a $\overline{\lambda}$-constacyclic code of length $np^s$ over $\overline{\mathcal{R}}.$ This, by Theorem \ref{p3}, implies that 
  \begin{equation}\label{ee0}
 \text{Depth}(\text{Tor}_{e-1}(\mathcal{C}))~=~\{\mathcal{S}_{1}(\mathcal{C})+1,\mathcal{S}_{1}(\mathcal{C})+2,\cdots, np^s\}.  
\end{equation}

We further note that each non-zero codeword  $c(x)\in \mathcal{C} \cap \langle \gamma^{e-1}\rangle$ can be written as $c(x)=\gamma^{e-1}c_1(x),$ where $c_1(x) \in \mathcal{R}_\lambda$ satisfies $\mu(c_1(x)) \neq 0.$ 
This, by Lemma \ref{l2}(b), implies that $\text{depth}(c(x))=\text{depth}(\mu(c_1(x))).$ This  gives
  \begin{equation}\label{ee1} \text{Depth}(\mathcal{C} \cap \langle \gamma^{e-1}\rangle)~=~ \text{Depth}(\text{Tor}_{e-1}(\mathcal{C}))~=~\{\mathcal{S}_{1}(\mathcal{C})+1,\mathcal{S}_{1}(\mathcal{C})+2,\cdots, np^s\}. \end{equation}
Next we assert that   \begin{equation}\label{ee2}
 \text{depth}(c(x))~\geq~\mathcal{S}_{1}(\mathcal{C})+1 \text{~~~for each~~~} c(x)\in \mathcal{C} \setminus \langle \gamma^{e-1}\rangle.
 \end{equation}

To prove this assertion,  let $0 \leq t \leq e-2$ be fixed, and let $c(x)\in \mathcal{C} \cap \big(\langle \gamma^{t}\rangle \setminus \langle \gamma^{t+1}\rangle\big).$ It is easy to see that the codeword $c(x)$ can be written as $c(x)=\gamma^{t}g(x),$ where $g(x)\in\mathcal{R}_{\lambda}$ satisfies $\mu(g(x)) \neq 0.$ This, by Lemma \ref{l2}(a), implies that $\text{depth}(c(x)) \geq  \text{depth}(\mu(g(x))).$ As $\mu(g(x)) \in \text{Tor}_{t}(\mathcal{C})\subseteq \text{Tor}_{e-1}(\mathcal{C}),$ by \eqref{ee0}, we see that $ \text{depth}(c(x)) \geq \text{depth}(\mu(g(x))) \geq \mathcal{S}_{1}(\mathcal{C})+1,$ which proves \eqref{ee2}. 

Now by \eqref{ee1} and \eqref{ee2}, the desired result follows. 
  \end{proof}

To illustrate the above theorem, we determine depth spectra of all 2-constacyclic codes of length $18$ over $\mathbb{Z}_9.$
\begin{ex} By Theorem \ref{Tstruc}(a), we see that all 2-constacyclic codes  of length $18$ over $\mathbb{Z}_9$ are given by $\mathcal{C}_{t}=\langle (x^2-8)^{t}\rangle,$ where  $0 \leq t \leq 18.$ Now by applying Theorem \ref{Tspec}, we have the following:
 \begin{center}\begin{tabular}{ll}
\begin{tabular}{ |c|c|c|c| } 
\hline
t & $|\mathcal{C}_t|$ &  $\text{Depth}(\mathcal{C}_t)$  \\
\hline
$0 \leq t \leq 9$ & $3^{36-2t}$ & $\{1,2,\cdots,18\}$  \\ 
 $t=10$& $3^{16}$ & $\{3,4,\cdots,18\}$  \\
$t=11$& $3^{14}$ & $\{5,6,\cdots,18\}$  \\
$t=12$& $3^{12}$ & $\{7,8,\cdots,18\}$  \\
$t=13$& $3^{10}$ & $\{9,10,\cdots,18\}$\\
\hline
\end{tabular}
& \begin{tabular}{ |c|c|c|c| } 
\hline
t & $|\mathcal{C}_t|$ &  $\text{Depth}(\mathcal{C}_t)$  \\
\hline
$t=14$ & $3^{8}$ & $\{11,12,\cdots,18\}$\\ 
$t=15$ & $3^{6}$ & $\{13,14,\cdots,18\}$\\ 
$t=16$ & $3^{4}$ & $\{15,16,17,18\}$  \\
$t=17$ & $3^{2}$ & $\{17,18\}$  \\
$t=18$ & $1$ & $\emptyset$  \\
\hline
\end{tabular}
 \end{tabular}
\end{center}  \qed \end{ex}
In the following theorem, we consider the case $\overline{\lambda}=1,$ and  we determine depth spectra of all $\lambda$-constacyclic codes of length $np^s$ over $\mathcal{R}.$   As $\lambda=\alpha+\gamma \beta$ with $\alpha \in \mathcal{T}$ and $\beta$ a unit in $ \mathcal{R},$ one can easily observe that $\overline{\lambda} =1$ if and only if  $\alpha =1,$ which holds if and only if $\lambda=1+\gamma \beta.$ When $\overline{\lambda}=1,$ without any loss of  generality, we can take $f_1(x)=x-1.$
  \begin{thm}\label{Tspec1} Let $\lambda= 1+\gamma \beta,$ where $\beta$ is a unit  in $\mathcal{R}.$ Let $\mathcal{C} =\langle (x-1)^{k_1} \prod\limits_{\ell=2}^{r} f_\ell(x)^{k_\ell} \rangle$ be a non-trivial $\lambda$-constacyclic code of length $np^s$ over $\mathcal{R}$ with the $i$th torsion code as $\text{Tor}_i(\mathcal{C})=\langle (x-1)^{\tau_1(i)}\prod\limits_{\ell=2}^{r}\overline{ f_\ell(x)}^{\tau_\ell(i)} \rangle$ for $0 \leq i \leq e-1,$  where $0 \leq k_\ell  \leq ep^s$ and $\tau_\ell(i)=\min\{(i+1)p^s, k_\ell\}-\min\{ i p^s, k_\ell\}$ for each $i$ and $\ell.$ 
Then the depth spectrum of the code $\mathcal{C}$ is given by
\begin{equation*}\text{Depth}(\mathcal{C})~=~\left\{\begin{array}{ll}
\{1,2,\cdots, np^s\}  &\text{if~ } 0 \leq k_1 <\max\{0,(e-n)p^s\};\\
  \{1,2,\cdots, ep^s-k_1\} \cup \{p^s+\mathcal{S}_{2}(\mathcal{C})+1,p^s+\mathcal{S}_{2}(\mathcal{C})+2,\cdots, np^s\} & \text{otherwise}.
\end{array}\right. \end{equation*}
\end{thm}
\begin{proof} 
To prove the result, we see that each non-zero codeword $c(x)\in \mathcal{C} \cap \langle \gamma^{e-1}\rangle$  can be expressed as $c(x)=\gamma^{e-1}c_1(x),$ where $c_1(x) \in \mathcal{R}_\lambda$ satisfies $\mu(c_1(x)) \neq 0.$  This, by Lemma \ref{l2}(b), implies that $\text{depth}(c(x))=\text{depth}(\mu(c_1(x))),$ which gives 
\begin{equation}\label{sub1}
\text{Depth}(\text{Tor}_{e-1}(\mathcal{C}))~=~ \text{Depth}(\mathcal{C} \cap \langle \gamma^{e-1} \rangle )~\subseteq ~\text{Depth}(\mathcal{C}).\end{equation} 
In  $\mathcal{R}_\lambda,$ it  is easy to see that $(x^n-1)^{p^s}=\gamma \mathcal{H}(x),$ where $\mathcal{H}(x)$ is a unit in $\mathcal{R}_\lambda.$ Let $\mathcal{G}(x)\in \mathcal{R}_\lambda$ satisfy  $\mathcal{G}(x)\mathcal{H}(x)=1$ in $\mathcal{R}_\lambda.$
 Now we shall consider the following three cases separately: (i) $0 \leq k_1 < \max\{0,(e-n)p^s\},$ (ii) $ \max\{0,(e-n)p^s\} \leq k_1< (e-1)p^s,$ and (iii) $(e-1)p^s \leq k_1 \leq ep^s.$  

\begin{description}

\item[(i)] Let $0 \leq k_1 <\max\{0,(e-n)p^s\}.$ Here we have $(e-n)p^s>0$ and $\tau_1{(e-1)}=0.$ This gives $\text{Tor}_{e-1}(\mathcal{C})=\langle \prod\limits_{\ell=2}^{r}\overline{ f_\ell(x)}^{\tau_\ell(e-1)} \rangle,$ which, by Proposition \ref{p2}, implies that   $$\text{Depth}(\text{Tor}_{e-1}(\mathcal{C}))~=~\{1,2,\cdots, p^s\}\cup\{p^s+\mathcal{S}_{2}(\mathcal{C})+1,p^s+\mathcal{S}_{2}(\mathcal{C})+2,\cdots, np^s\}.  $$ This, by \eqref{sub1}, further implies that \begin{equation}\label{e4}
\{1,2,\cdots, p^s\}\cup\{p^s+\mathcal{S}_{2}(\mathcal{C})+1,p^s+\mathcal{S}_{2}(\mathcal{C})+2,\cdots, np^s\} ~\subseteq~\text{Depth}(\mathcal{C}).  \end{equation}
For $0\leq u \leq (n-1)p^s-1,$ let us define
 \begin{equation*}  c_u(x)~=~(-1)^{(e-1)p^s}(1-x)^{(e-n)p^s+u}\prod\limits_{\ell=2}^{r}\overline{f_\ell(x)}^{ep^s}\mathcal{G}(x)^{e-1} . \end{equation*}
Now as $k_1 <(e-n)p^s,$ we note that $c_u(x) \in \mathcal{C}$ for $0\leq u \leq (n-1)p^s-1.$  This implies that   \begin{equation*} (1-x)^{(n-1)p^s-u}c_u(x)~=~(x-1)^{(e-1)p^s}\prod\limits_{\ell=2}^{r}f_\ell(x)^{ep^s}\mathcal{G}(x)^{e-1}~=~\gamma^{e-1}\prod\limits_{\ell=2}^{r}f_\ell(x)^{p^s}\in\mathcal{C} \text{~~~for~~~} 0 \leq u \leq (n-1)p^s-1.\end{equation*}   Further, by Lemma \ref{l2}(b), we see that   \begin{equation*} \text{depth}\big( (1-x)^{(n-1)p^s-u}c_u(x)\big)~=~\text{depth}(\gamma^{e-1}\prod\limits_{\ell=2}^{r}f_\ell(x)^{p^s})~=~\text{depth}(\prod\limits_{\ell=2}^{r}\overline{f_\ell(x)}^{p^s}) \text{~~~for~each~~} u.  \end{equation*} 
We also observe that $(1-x)^{p^s}\prod\limits_{\ell=2}^{r}\overline{f_\ell(x)}^{p^s}=0$ in $\overline{\mathcal{R}}_\lambda$ and $ \text{deg }\big((1-x)^{p^s-1}\prod\limits_{\ell=2}^{r}\overline{f_\ell(x)}^{p^s}\big) =np^s-1.$ This, by Proposition \ref{p1}, implies that   \begin{equation*} D^{p^s-1}(\prod\limits_{\ell=2}^{r}\overline{f_\ell(x)}^{p^s}) ~\neq ~(0,0,\cdots,0)\in \mathcal{R}^{(n-1)p^s+1} \text{~~and~~} D^{p^s}(\prod\limits_{\ell=2}^{r}\overline{f_\ell(x)}^{p^s})~ =~(0,0,\cdots,0)\in \mathcal{R}^{(n-1)p^s},\end{equation*}  which gives    \begin{equation*}  \text{depth}\big((1-x)^{(n-1)p^s-u}c_u(x)\big)~=~\text{depth}(\prod\limits_{\ell=2}^{r}\overline{f_\ell(x)}^{p^s})~=~p^s.\end{equation*} 
Now by applying Lemma \ref{ll3}, we obtain $\text{depth}(c_u(x))=(n-1)p^s-u+p^s$ for $0\leq u \leq (n-1)p^s-1.$  This implies that   \begin{equation*} \{p^s+1,p^s+2,\cdots,np^s\} ~\subseteq~ \text{Depth}(\mathcal{C}).\end{equation*}  From this and by \eqref{e4},  we obtain $\text{Depth}(\mathcal{C})=\{1,2,\cdots, np^s\}.$  

 \item[(ii)] Next let $\max\{0,(e-n)p^s\} \leq k_1 <(e-1)p^s.$ Here we have $\tau_1(e-1)=0,$ which gives $\text{Tor}_{e-1}(\mathcal{C})=\langle \prod\limits_{i=2}^{r} \overline{ f_i(x)}^{\tau_i(e-1)} \rangle.$ This, by Proposition \ref{p2}, implies that   \begin{equation*} \text{Depth}(\text{Tor}_{e-1}(\mathcal{C}))~=~\{1,2,\cdots, p^s\}\cup\{p^s+\mathcal{S}_{2}(\mathcal{C})+1,p^s+\mathcal{S}_{2}(\mathcal{C})+2,\cdots, np^s\}.\end{equation*} 
 From  this and by \eqref{sub1}, we get   \begin{equation}\label{e5}
\text{Depth}(\mathcal{C} \cap \langle \gamma^{e-1} \rangle )~ =~\{1,2,\cdots, p^s\}\cup\{p^s+\mathcal{S}_{2}(\mathcal{C})+1,p^s+\mathcal{S}_{2}(\mathcal{C})+2,\cdots, np^s\}~\subseteq~\text{Depth}(\mathcal{C}).  \end{equation}

For $0\leq u < (e-1)p^s-k_1,$ let us define
  \begin{equation*}  a_u(x)~=~(-1)^{(e-1)p^s}(1-x)^{k_1+u}\prod\limits_{\ell=2}^{r}f_\ell(x)^{ep^s}\mathcal{G}(x)^{e-1} .  \end{equation*}

Further, we note that $a_u(x)\in \mathcal{C},$  which implies that   \begin{equation*}  (1-x)^{(e-1)p^s-k_1-u}a_u(x)~=~(x-1)^{(e-1)p^s}\prod\limits_{\ell=2}^{r}f_\ell(x)^{ep^s}\mathcal{G}(x)^{e-1}~=~\gamma^{e-1}\prod\limits_{\ell=2}^{r}f_\ell(x)^{p^s}\in\mathcal{C}.  \end{equation*}   Next by Lemma \ref{l2}(b), we see that   \begin{equation*}  \text{depth}\big((1-x)^{(e-1)p^s-k_1-u}a_u(x)\big)~=~\text{depth}(\gamma^{e-1}\prod\limits_{\ell=2}^{r}f_\ell(x)^{p^s} )~=~\text{depth}(\prod\limits_{\ell=2}^{r}\overline{f_\ell(x)}^{p^s}) \text{~~for~each~~}u.  \end{equation*}  Further, by Proposition \ref{p1}, we see that $\text{depth}\big((1-x)^{(e-1)p^s-k_1-u}a_u(x)\big)=\text{depth}(\prod\limits_{\ell=2}^{r}\overline{f_\ell(x)}^{p^s})=p^s,$ which, by Lemma \ref{ll3}, implies that $\text{depth}(a_u(x))=(e-1)p^s-k_1-u+p^s$ for $0\leq u < (e-1)p^s-k_1.$  This gives $\{p^s+1,p^s+2,\cdots,ep^s-k_1\} \subseteq \text{Depth}(\mathcal{C}) .$ From this and by \eqref{e5}, we see that \begin{equation}\label{e6}\{1,2,\cdots, ep^s-k_1\}\cup \{p^s+\mathcal{S}_{2}(\mathcal{C})+1,p^s+\mathcal{S}_{2}(\mathcal{C})+2,\cdots, np^s\} ~\subseteq~\text{Depth}(\mathcal{C}).  \end{equation}

Now we assert that \begin{equation}\label{e7}\text{Depth}(\mathcal{C})~\subseteq~ \{1,2,\cdots, ep^s-k_1\}\cup\{p^s+\mathcal{S}_{2}(\mathcal{C})+1,p^s+\mathcal{S}_{2}(\mathcal{C})+2,\cdots, np^s\}.  \end{equation}

When $ep^s-k_1\geq p^s+\mathcal{S}_{2}(\mathcal{C})+1,$ we see that $$\{1,2,\cdots, ep^s-k_1\}\cup\{p^s+\mathcal{S}_{2}(\mathcal{C})+1,p^s+\mathcal{S}_{2}(\mathcal{C})+2,\cdots, np^s\}~=~\{1,2,\cdots, np^s\},$$ and hence \eqref{e7} holds in this case.  So from now on, we assume that $ep^s-k_1< p^s+\mathcal{S}_{2}(\mathcal{C})+1.$ To prove the assertion \eqref{e7}, it suffices to prove the following:
 \begin{equation}\label{e10} \text{either~~~} \text{depth}(c(x))~ \leq ~ep^s-k_1 \text{~~~or~~~} \text{depth}(c(x)) ~\geq ~p^s+\mathcal{S}_{2}(\mathcal{C})+1  \text{~~~for~each~~~} c(x)(\neq 0) \in \mathcal{C}.  \end{equation} 
 
To do this, we note that each non-zero codeword $c(x)\in \mathcal{C}$  can be written as $c(x)=(1-x)^{h_1}\prod\limits_{\ell=2}^{r}f_\ell(x)^{h_\ell}h(x),$ where $h_\ell \geq k_\ell $  for $1 \leq \ell \leq r,$ and $h(x) \in \mathcal{R}[x]$ is such that $\mu(h(x)) \neq 0$ and $h(x)$ is coprime to  $x^n-1$ in $\mathcal{R}[x].$ Now the following two cases arise: \textbf{A.} $h_\ell \geq ep^s$ for  $2 \leq \ell \leq r,$ and \textbf{B.} there exists an integer $t$ satisfying $2 \leq t \leq r$ and  $h_t <ep^s.$
\\\textbf{A.} Let us suppose that $h_\ell \geq ep^s$ for $2 \leq \ell \leq r.$ As $c(x) \in \mathcal{C}$ is a non-zero codeword, we must have $h_1<ep^s.$  We further note that   \begin{equation*} (1-x)^{ep^s-h_1}c(x)~=~(1-x)^{ep^s}\prod\limits_{\ell=2}^{r}f_\ell(x)^{h_\ell}h(x)~=~0 \text{~~~in~~~} \mathcal{R}_\lambda.  \end{equation*}  This implies that $D^{ep^s-h_1}(c(x))=(0,0,\cdots,0)\in \mathcal{R}^{np^s-ep^s+h_1},$ which further implies that $\text{depth}(c(x))\leq ep^s-h_1\leq ep^s-k_1.$ This proves \eqref{e10} in this case.

\textbf{B.} Now suppose that there exists an integer $t$ satisfying  $2 \leq t\leq r$ and  $h_t <ep^s.$ As $h(x)$ is coprime to $x^n-1$ in $\mathcal{R}[x],$  we see that $h(x)$ is a unit in $\mathcal{R}_\lambda.$ Further, since $h_t< ep^s,$ we note that  \begin{equation*} (1-x)^{p^s+\mathcal{S}_{2}(\mathcal{C})}  c(x)~=~(1-x)^{p^s+\mathcal{S}_2(\mathcal{C})+h_1}\prod\limits_{\ell=2}^{r}f_\ell(x)^{h_\ell} h(x)\not\in\langle (x^n-1)^{ep^s}\rangle =\{0\}\text{~~~in~~~}\mathcal{R}_\lambda. \end{equation*} That is, we have $(1-x)^{p^s+\mathcal{S}_{2}(\mathcal{C})}  c(x)~\neq~ 0$ in $\mathcal{R}_\lambda.$   

Now as $p^s+\mathcal{S}_2(\mathcal{C})+h_1 \geq p^s+\mathcal{S}_{2}(\mathcal{C})+k_1 \geq ep^s,$ there exist integers $u$ and $v$ satisfying $2 \leq u\leq r,$ $0 \leq v\leq e-1,$  $h_u <(v+1)p^s$ and $h_\ell \geq vp^s$ for $1 \leq \ell \leq r.$ So we can write   \begin{equation*} (1-x)^{p^s+\mathcal{S}_2(\mathcal{C})+h_1}\prod\limits_{\ell=2}^{r}f_\ell(x)^{h_\ell}~=~(-1)^{vp^s}\gamma^v (1-x)^{p^s+\mathcal{S}_2(\mathcal{C})+h_1-vp^s}\prod\limits_{\ell=2}^{r}f_\ell(x)^{h_\ell-vp^s}\mathcal{H}(x)^{v}.  \end{equation*}

Next we assert that
 \begin{equation}\label{e12}
\text{deg }\Big((1-x)^{p^s+\mathcal{S}_{2}(\mathcal{C})}c(x)\Big) ~\geq~ p^s+\mathcal{S}_{2}(\mathcal{C}).
\end{equation}
To prove this assertion, we note that   \begin{equation*} (1-x)^{p^s+\mathcal{S}_{2}(\mathcal{C})}c(x)~=~(1-x)^{h_1+p^s+\mathcal{S}_{2}(\mathcal{C})}\prod\limits_{\ell=2}^{r}f_\ell(x)^{h_\ell}h(x)\in\langle (x-1)^{ep^s}\prod\limits_{\ell=2}^{r}f_\ell(x)^{k_\ell} \rangle~=~\mathcal{C}_{1}\text{~(say)}.\end{equation*}  Let us take   \begin{equation*} A(x)~=~(-1)^{vp^s} (1-x)^{p^s+\mathcal{S}_2(\mathcal{C})+h_1-vp^s}\prod\limits_{\ell=2}^{r}f_\ell(x)^{h_\ell-vp^s}\mathcal{H}(x)^{v}h(x) \end{equation*}  so that  $(1-x)^{p^s+\mathcal{S}_{2}(\mathcal{C})}c(x)=\gamma^v A(x).$ Note that $\mu(A(x))\neq 0.$ Now as \begin{equation*}(1-x)^{p^s+\mathcal{S}_{2}(\mathcal{C})}c(x) ~=~\gamma^v A(x)\in\mathcal{C}_{1},\end{equation*} we see, by Theorem \ref{Tstruc}(b), that   \begin{equation*} \mu(A(x))\in \text{Tor}_{e-1}(\mathcal{C}_{1})~=~\langle (x-1)^{p^s} \prod\limits_{\ell=2}^{r}\overline{ f_\ell(x)}^{\tau_\ell(e-1)}\rangle. \end{equation*} This implies that $\text{deg }(\mu(A(x))) \geq p^s+\mathcal{S}_{2}(\mathcal{C}),$ which further implies that
\begin{equation*}  \text{deg }\Big((1-x)^{p^s+\mathcal{S}_{2}(\mathcal{C})}c(x)\Big)~=~\text{deg }(\gamma^vA(x))~\geq ~\text{deg }(\mu(A(x))) ~\geq~ p^s+\mathcal{S}_{2}(\mathcal{C}).\end{equation*} 
Further, by applying Proposition \ref{p1} and by \eqref{e12}, we get  \begin{equation*} D^{p^s+\mathcal{S}_{2}(\mathcal{C})}(c(x)) ~\neq~ (0,0,\cdots,0)\in\mathcal{R}^{(n-1)p^s-\mathcal{S}_{2}(\mathcal{C})}. \end{equation*}  This implies that $\text{depth }(c(x)) \geq p^s+\mathcal{S}_{2}(\mathcal{C})+1,$ which proves  \eqref{e10}. Now by \eqref{e7} and \eqref{e10}, we get the desired result.
\item[(iii)] Finally, let $(e-1)p^s \leq k_1 \leq ep^s.$ Here for $0 \leq i \leq e-2,$ we have $\tau_{1}(i)=p^s,$ which gives    \begin{equation*} \text{Tor}_{i}(\mathcal{C})~=~\bigg\langle (x-1)^{p^s}\prod\limits_{\ell=2}^{r} \overline{ f_\ell(x)}^{\tau_\ell(i)}\bigg\rangle. \end{equation*}  We also note that  $\tau_1(e-1)=k_1-(e-1)p^s,$ which implies that
  \begin{equation*} \text{Tor}_{e-1}(\mathcal{C})~=~\bigg\langle (x-1)^{k_1-(e-1)p^s}\prod\limits_{\ell=2}^{r} \overline{ f_\ell(x)}^{\tau_\ell(e-1)} \bigg\rangle.  \end{equation*}  This, by Proposition \ref{p2}, implies that
   \begin{equation*} \text{Depth}(\text{Tor}_{i}(\mathcal{C}))~=~\Big\{p^s+\sum\limits_{\ell=2}^{r}d_\ell\tau_\ell(i)+1,p^s+\sum\limits_{\ell=2}^{r}d_\ell\tau_\ell(i)+2,\cdots, np^s\Big\} \text{~~~for~~~} 0 \leq i \leq e-2 \end{equation*}   and that   \begin{equation*} \text{Depth}(\text{Tor}_{e-1}(\mathcal{C}))~=~\{1,2,\cdots, ep^s-k_1\}\cup\{p^s+\mathcal{S}_{2}(\mathcal{C})+1,p^s+\mathcal{S}_{2}(\mathcal{C})+2,\cdots, np^s\}.\end{equation*} 
 Now by \eqref{sub1}, we get  \begin{equation}\label{ee11} \text{Depth}(\mathcal{C} \cap \langle \gamma^{e-1}\rangle)~=~\{1,2,\cdots, ep^s-k_1\}\cup\{p^s+\mathcal{S}_{2}(\mathcal{C})+1,p^s+\mathcal{S}_{2}(\mathcal{C})+2,\cdots, np^s\} ~\subseteq~ \text{Depth}(\mathcal{C}).\end{equation}
Next we assert that   \begin{equation}\label{ee12}
 \text{depth}(c(x)) ~\geq~   p^s+\mathcal{S}_{2}(\mathcal{C})+1 \text{~~~for~each~~~} c(x)\in \mathcal{C} \setminus \langle \gamma^{e-1}\rangle.
 \end{equation}
To prove this assertion,  let $0 \leq t \leq e-2$ be fixed.  We note that each $c(x)\in \mathcal{C} \cap  \big(\langle\gamma^{t}\rangle \setminus \langle \gamma^{t+1}\rangle\big)$ can be written as $c(x)=\gamma^{t}g(x),$ where $g(x)\in\mathcal{R}_{\lambda}$ satisfies $\mu(g(x)) \neq 0.$ This, by Lemma \ref{l2}(a), implies that $\text{depth}(c(x)) \geq  \text{depth}(\mu(g(x))).$ As $\mu(g(x)) \in \text{Tor}_{t}(\mathcal{C}),$ we have   \begin{equation}\label{s1}
 \text{depth}(c(x))~\geq~  \text{depth}(\mu(g(x))) ~\geq~ p^s+ \sum\limits_{\ell=2}^{r}d_\ell\tau_\ell(t)+1.
 \end{equation} Since $\text{Tor}_{t}(\mathcal{C}) \subseteq \text{Tor}_{e-1}(\mathcal{C}),$ we see that $ (x-1)^{k_1-(e-1)p^s}\prod\limits_{\ell=2}^{r} \overline{ f_\ell(x)}^{\tau_\ell(e-1)}$ divides $ (x-1)^{p^s}\prod\limits_{\ell=2}^{r} \overline{ f_\ell(x)}^{\tau_\ell(t)}$ in $\overline{\mathcal{R}}[x],$ which implies that   \begin{equation*} \sum\limits_{\ell=2}^{r}d_\ell\tau_\ell(t)~=~\text{deg }\Big(\prod\limits_{\ell=2}^{r} \overline{ f_\ell(x)}^{\tau_\ell(t)}\Big)~\geq~ \text{deg }\Big(\prod\limits_{\ell=2}^{r} \overline{ f_\ell(x)}^{\tau_\ell(e-1)}\Big)~=~\sum\limits_{\ell=2}^{r} d_{\ell} \tau_{\ell}(e-1)~=~\mathcal{S}_2(\mathcal{C}).  \end{equation*} 

This, by \eqref{s1}, implies that   \begin{equation*} \text{depth}(c(x))~\geq~ p^s+\sum\limits_{\ell=2}^{r}d_\ell\tau_\ell(t)+1 ~ \geq~ p^s+\mathcal{S}_2(\mathcal{C})+1,  \end{equation*} which  proves  \eqref{ee12}. Now the desired result follows immediately from \eqref{ee11} and \eqref{ee12}. 
\end{description}This completes the proof of the theorem.
\end{proof}
To illustrate the above theorem, we determine depth spectra of some negacyclic codes of length $56$ over $GR(4,4)$ in the following example. 
 \begin{ex} By Theorem \ref{Tstruc}(a),  we see that all negacyclic codes of length $56$ over $GR(4,4)$ are given by $\mathcal{C}_{k_1,k_2,k_3}=\langle (x+3)^{k_1}(x^3 + 2x^2 + x + 3)^{k_2}( x^3 + 3x^2 + 2x + 3)^{k_3}\rangle,$ where  $0 \leq k_1,k_2,k_3 \leq 2^4.$ Now by applying Theorem \ref{Tspec1}, we have the following:
  \begin{center}\hspace{-6.9mm}\begin{tabular}{ll}
 \begin{tabular}{ |c|c|c|c| } 
\hline
$(k_1,k_2,k_3)$ & $|\mathcal{C}_{k_1,k_2,k_3}|$ & $\text{Depth}(\mathcal{C}_{k_1,k_2,k_3})$  \\
\hline
$(14,12,13)$ & $2^{92}$ & $ \{1,2\}\cup\{36,37,\cdots,56\}$  \\ 
 $(14,14,11)$ & $2^{92}$ & $ \{1,2\}\cup\{36,37,\cdots,56\}$  \\ 
$(14,16,9)$ & $2^{92}$ & $ \{1,2\} \cup\{36,37,\cdots,56\}$  \\ 
$(14,10,15)$ & $2^{92}$ & $ \{1,2\}\cup\{36,37,\cdots,56\}$  \\ 
$(15,16,5)$ & $2^{136}$ & $ \{1\} \cup\{33,34,\cdots,56\}$  \\ 
$(15,6,16)$ & $2^{124}$ & $ \{1\}\cup\{33,34,\cdots,56\}$  \\ 
\hline
\end{tabular}
& \begin{tabular}{ |c|c|c|c| } 
\hline
$(k_1,k_2,k_3)$ & $|\mathcal{C}_{k_1,k_2,k_3}|$ & $\text{Depth}(\mathcal{C}_{k_1,k_2,k_3})$  \\
\hline
$(7,6,5) $ & $2^{288}$ & $\{1,2,\cdots,56\}$  \\ 
$(7,3,4)$ & $2^{336}$ & $\{1,2,\cdots,56\}$  \\ 
$(16,5,16)$ & $2^{132}$ & $\{33,37,\cdots,56\}$  \\ 
$(10,6,17)$ & $2^{132}$ & $ \{1,2,\cdots,6\} \cup\{36,37,\cdots,56\}$  \\ 
$(4,9,10)$& $2^{204}$ & $\{1,2,\cdots,12\}\cup\{18,19,\cdots,56\}$  \\
$(13,6,10)$& $2^{204}$ & $\{1,2,3\}\cup\{15,16,\cdots,56\}$  \\ 
\hline \end{tabular}\end{tabular} \end{center}\qed  \end{ex}
  \section{Conclusion and Future work}\label{con}  
Let $\mathcal{R}$ be a finite commutative chain ring with the unique maximal ideal as $\langle \gamma \rangle.$  In this paper,  depth spectra of all repeated-root $(\alpha+\gamma \beta)$-constacyclic codes of arbitrary lengths over $\mathcal{R}$ are explicitly determined, where $\alpha$ is a non-zero element of the Teichm\"{u}ller set of $\mathcal{R}$ and $\beta $ is a unit in  $\mathcal{R}.$ 

It would be interesting  to determine depth distributions of all constacyclic codes of arbitrary lengths over $\mathcal{R}.$ It would be of great interest to explore more applications of  depth distributions of linear codes over finite rings in game theory,  communication theory and cryptography.

  \section{Acknowledgment}\label{acknow}   
I am grateful to my Ph.D. supervisor, Dr. Anuradha Sharma, for discussions and her help in preparing this manuscript.

\end{document}